\newcommand{\columnsversion}[2]{#1}{} 
\newcommand{\comment}[1]{}

\columnsversion{
\documentclass[11pt,letterpaper]{article}

\topmargin=-0.35in \topskip=0pt \headsep=15pt \oddsidemargin=0pt
\textheight=9in \textwidth=6.52in \voffset=0in

\usepackage{amsmath,amsfonts,graphpap,amscd,mathrsfs,graphicx,lscape}
\usepackage{epsfig,amssymb,amstext,xspace}
\usepackage{theorem, bbm}
\usepackage{algorithm,algorithmic}

}{
\documentclass[twoside,leqno,twocolumn]{article}  
\usepackage{ltexpprt} 

\usepackage{amsmath,amsfonts,graphpap,amscd,mathrsfs,graphicx,lscape}
\usepackage{epsfig,amssymb,amstext,xspace}
\usepackage{bbm}
\usepackage{algorithm,algorithmic}

}

\usepackage{color}              
\usepackage{epsfig}
\usepackage{paralist}
\usepackage{multirow}

\usepackage{caption}
\newcommand{\abs}[1]{\vert{#1}\vert}

\columnsversion{
\newtheorem{theorem}{Theorem}[section]
\newtheorem{lemma}[theorem]{Lemma}
\newtheorem{metalemma}[theorem]{Meta-Lemma}

\newtheorem{prop}[theorem]{Proposition}

\newtheorem{cor}[theorem]{Corollary}

\newtheorem{example}[theorem]{Example}
\newtheorem{defn}[theorem]{Definition}

{\theorembodyfont{\rmfamily} }
}{
\newtheorem{metalemma}[@theorem]{Meta-Lemma}
\newtheorem{defn}[@theorem]{Definition}
\newtheorem{cor}[@theorem]{Corollary}
\newtheorem{prop}[@theorem]{Proposition}

}

\columnsversion{
\newenvironment{proof}{\bg{Proof : }}{\ed}
\newenvironment{warmup}{\bg{Warm Up: }}{\ed}
\def\squareforqed{\hbox{\rule{2.5mm}{2.5mm}}}
\def\blackslug{\rule{2.5mm}{2.5mm}}

\def\QED{\ifmmode\squareforqed 
  \else{\nobreak\hfil   
    \penalty50                 
    \hskip1em                  
    \null                      
    \nobreak                   
    \hfil                      
    \squareforqed              
    \parfillskip=0pt           
    \finalhyphendemerits=0     
    \endgraf}                  
  \fi}

\def\blksquare{\rule{2mm}{2mm}}
\def\qedsymbol{\blksquare}
\newcommand{\bg}[1]{\medskip\noindent{\bf #1}}
\newcommand{\ed}{{\hfill\qedsymbol}\medskip}
}{
\newenvironment{warmup}{{\it Warm Up: }}{\ed}
}

\newenvironment{proofof}[1]{{\it{Proof of #1 : }}}{\ed}

\newcommand{\R}{\ensuremath{\mathbb R}}

\newcommand{\F}{\ensuremath{\mathcal F}}

\columnsversion{
\newcommand{\email}[1]{\texttt{#1}}
}{}

\newcommand{\base}{\mathtt{b}}
\newcommand{\aug}{\mathtt{a}}

\begin{document}

\columnsversion{
\setcounter{page}{0}
}{}

\title{Clinching Auctions with Online Supply}
\columnsversion{
\author{
Gagan Goel\\
       Google Inc., New York\\
       \email{gagangoel@google.com}
\and
Vahab Mirrokni\\
       Google Inc., New York\\
       \email{mirrokni@google.com}
\and
Renato Paes Leme\\
       Cornell University\\
       \email{renatoppl@cs.cornell.edu}
}}{
\author{Gagan Goel\thanks{ {\tt gagangoel@google.com}, Google Inc., New York.}
\\ \and
Vahab Mirrokni\thanks{ {\tt mirrokni@google.com}, Google Inc., New York } \\
\and
Renato Paes Leme\thanks{ {\tt renatoppl@cs.cornell.edu}, Dept of Computer
Science, Cornell University. This work was done while the author was a summer
intern at Google NYC. During the rest of the academic year, he is supported by
a Microsoft Research Fellowship.}}}
\date{}

\maketitle

\begin{abstract}
Auctions for perishable goods such as internet ad inventory need to
make real-time allocation and pricing decisions as the supply of the
good arrives in an online manner, without knowing the entire supply in
advance. These allocation and pricing decisions get complicated when
buyers have some global constraints. In this work, we consider a
multi-unit model where buyers have global {\em budget} constraints,
and the supply arrives in an online manner. Our main contribution is
to show that for this setting there is an individually-rational,
incentive-compatible and Pareto-optimal auction that allocates these
units and calculates prices on the fly, without knowledge of the total
supply. We do so by showing that the Adaptive Clinching Auction
satisfies a {\em supply-monotonicity} property.

We also analyze and discuss, using examples, how the insights gained
by the allocation and payment rule can be applied to design better ad
allocation heuristics in practice. Finally, while our main technical
result concerns multi-unit supply, we propose a formal model of online
supply that captures scenarios beyond multi-unit supply and has
applications to sponsored search. We conjecture that our results for
multi-unit auctions can be extended to these more general models.
\end{abstract}

\columnsversion{\newpage}{}

\section{Introduction}
The problem of selling advertisement on the web is essentially an online
problem - the supply (pageviews) arrives dynamically and decisions on how to
allocate ads to pageviews and price these ads need to be taken instantaneously,
without full knowledge of the future supply. What makes these decisions complex
is the fact that buyers have budget constraints, which ties the allocation and
pricing decisions across different time steps. Another complicating feature of the
 online advertisement markets is that buyers are strategic and can misreport their values to
their own advantage.

These observations have sparkled a fruitful line of research in two different
directions. First is that of designing online algorithms where one assumes that
the supply is coming online, but makes a simplifying assumption that buyers are
non-strategic. This line of research has led to novel tools and techniques in
the design of online algorithms (see for example ~\cite{jacmMehtaSVV07,BJN07, DH09, AggarwalGKM11}). 
The second line of
research considers the design of incentive-compatible mechanisms assuming that
buyers are strategic, but makes an assumption that the supply is known
beforehand. Handling budget constraints using truthful mechanisms is
non-trivial since standard VCG-like techniques fail when the player utilities
are not quasi-linear. In a seminal work, Dobzinski, Lavi and Nisan
\cite{dobzinski12} showed that one can adapt Ausubel's clinching auction
\cite{Ausubel_multi} to achieve Pareto-optimal outcomes for the case of
multi-unit supply. In settings with budget constraints, the goal of maximizing
social welfare is unattainable and efficiency is achieved through
Pareto-optimal outcomes. In fact, if budgets are sufficiently large,
Pareto-optimal outcomes are exactly the ones that maximize social welfare
\cite{goel12}.

From a practical standpoint, it is important to understand what can be done when
both the above scenarios are present at the same time. Motivated by this, we
study the following question in this paper:
{\em Can one design efficient incentive-compatible mechanisms for the
case when agents have budget constraints and the supply arrives online?}

A closely related question was studied by Babaioff, Blumrosen and Roth
\cite{babaioff10}, who asked weather it was possible to obtain efficient
incentive-compatible mechanisms with online supply, but instead of budget
constraints, they considered capacity constraints, i.e., each agent wants at
most $k$ items (capacity) rather than having at most $B$ dollars to spend
(budget). They showed that no such mechanism can be efficient and proved
lower bounds on the efficiency that could be achieved.

Such lower bounds seem to offer a grim perspective on what can be done with
budget constraints, since typically, budget constraints are less
well-behaved than capacity constraints. On the contrary, and somewhat
surprisingly, we show that for budget constraints it is
possible to obtain incentive compatible and Pareto-optimal auctions that
allocate and charge for items as they arrive, by showing that the Adaptive
Clinching Auction in  \cite{dobzinski12} for multi-unit supply can be implemented in an online
manner. More formally, we show that the clinching auction for the multi-unit
supply case satisfies the following {\em supply-monotonicity} property: Given
the allocation and payments obtained by running the auction for initial supply
$s$, one can obtain the allocation and payments for any other supply $s' \geq s$
by {\em augmenting} to the auction outcome for supply $s$. In other words,
it is possible to find an allocation for the extra $s'-s$ items and extra
(non-negative) payments such that when added to clinching auction outcome for
the supply $s$, we obtain the clinching auction outcome for supply $s'$.
Moreover, we show that each agent's utility is also monotone with respect to the
supply, i.e., agents do not have incentive to leave the auction prematurely. 

From a technical perspective, proving the above result requires a deeper
understanding of the structure of the clinching auction, which in general is
difficult to analyze because it is described using a differential ascending
price procedure rather than a one-shot outcome like VCG. In order to do so, we
study the description of the clinching auction given by Bhattacharya, Conitzer,
Munagala and Xia \cite{Bhattacharya10} by means of a differential equation. At
its heart, the proof of the supply monotonicity is a coupling argument. We
analyze two parallel differential procedures whose limits correspond to the
outcome of the clinching auction with the same values and budgets but different
initial supplies. We prove that either one stays ahead of the other or they meet
and from this point on they evolve identically (for carefully chosen concepts
of `stay ahead' and `meet'). We identify many different invariants in the
differential description of the auction, and use tools from real analysis to
show that these invariants hold.\columnsversion{}{\\}

\paragraph{Towards better heuristics for ad-allocation.}
One of the main goals of this research program is to provide insights for the
design of better heuristics to deal with budget-constrained agents in
real ad auctions. Most heuristics in practice are based on bid-throttling or
bid-lowering. Bid-throttling probabilistically removes a player from the
auction based on her spent budget (throttling). Bid-lowering runs a standard second
price auction with modified bids. While sound from an algorithmic perspective,
bid-throttling and bid-lowering are not integrated with the underlying auction from the perspective of
incentives.
We believe clinching auctions provide better insights into designing heuristics that are
more robust to strategic behavior. Towards this goal, we analyze the online
allocation rule obtained from the clinching auction and provide a qualitative
description of how allocation and payments evolve when new items arrive, and show how this description is significantly different from the bid-throttling or bid-lowering heuristics that are applied in practice.
In clinching, the fact that an agent got many
items for an expensive price in the beginning (once the items were scarce) gives
him an advantage over items in the future if/once they become abundant, i.e.,
this agent will have the possibility of acquiring these items for a lower price
than the other agents.\columnsversion{}{\\}

\paragraph{Online supply beyond Multi Unit Auctions.}
Since the groundbreaking work of Dobzinski, Lavi and Nisan \cite{dobzinski12},
clinching auctions have been extended beyond multi-unit auctions in the
offline-supply setting: first to matching markets by Fiat, Leonardi, Saia
and Sankowski \cite{fiat_clinching}, then to sponsored search setting by
Colini-Baldeschi, Henzinger, Leonardi and Starnberger \cite{henzinger11} and
to general polymatroids by Goel, Mirrokni and Paes Leme \cite{goel12}. It is
natural to ask if such offline-auctions have online supply counterparts.

We leave this question as the main open problem in this paper. However, before one formulates this question, we need to define what we mean by
online supply in such settings. One contribution of this paper is to define a
{\em model of online supply} for allocation constraints beyond multi-units.
Using our definition of online supply for generic constraints and the
definition of AdWords Polytope in \cite{goel12}, we can capture for example:
(1) sponsored search with multiple slots, where for each pageview we need to
decide on an allocation of agents to slots and (ii) matching constraints: each
arriving pageview can only be allocated to a subset of
advertisers.\columnsversion{}{\\}

\paragraph{Related Work.}

The study of auctions with online supply was initiated in Mahdian and Saberi
\cite{Mahdian2006} who study multi-unit auctions with the objective of
maximizing revenue. They provide a constant competitive auction with the
optimal offline single-price revenue. Devanur and Hartline \cite{devanur09}
study this problem in both the Bayesian and prior-free model. In the Bayesian model,
they argue that there is no separation between the online and offline problem.
This discussion is then extended to the prior-free setting. The results in
\cite{devanur09} assume that the payments can be deferred until all supply is
realized, while allocation needs to be done online.

Our work is more closely related to the work by Babaioff, Blumrosen and Roth
\cite{babaioff10}, which study the online supply model with the goal of
maximizing social welfare. Unlike previous work, they insist (as we also do)
that payments are charged in an online manner.  
This is a desirable property from a practical standpoint, since it
allows players to monitor their spend in real-time.
Their results are mainly
negative: they prove lower bounds on the approximability of social welfare in
setting where the supply is online. Efficiency is only recovered when stochastic
information on the supply is available.

We should also note that there is a long line of research at the intersection
of online algorithms and mechanism design, mostly dealing with agents arriving
and departing in an online manner. We refer to Parkes \cite{parkes07a} for a
survey.

Another stream of related works comes from the literature on mechanism design
with budget constrained agent. This line of inquire was initiated in Dobzinski
et al \cite{dobzinski12}, who proposed a mechanism based on Ausubel's celebrated
clinching framework \cite{Ausubel_multi}. The authors propose a mechanism for
multi-unit auctions and indivisible goods and a mechanism for $2$ players and
divisible goods called the {\em Adaptive Clinching Auction}. The mechanism is
extended to $n$ players by Bhattacharya et al \cite{Bhattacharya10} by means of
a differential ascending process whose limit is the allocation and payments of
the Adaptive Clinching Auction. The authors also show how to use
randomization to enable the auction to handle private budgets. Many subsequent
papers deal with extending the clinching auctions to more general environments
beyond multi-units: Fiat el al \cite{fiat_clinching}, Colini-Baldeschi et al
\cite{henzinger11} and Goel et al \cite{goel12}.

\section{Preliminary Definitions}

An auction is defined by a set of $n$ {\em players} equipped with utility
functions and an {\em environment}, which specifies the set of feasible
allocations. Formally, we consider a divisible\footnote{our decision of
considering divisible goods is motivated by our application. In sponsored
search, the number of items (pageviews) arriving at each time is enormous,
making fractional allocations essentially feasible.} good $g$ such
that an allocation of this good will be represented by a vector $x = (x_1,
\hdots, x_n)$, meaning that player $i$ got $x_i$ units of the good. Player $i$
has a set of private types $\Theta_i$ and his utility function will depend on
his type $\theta_i$, the amount $x_i$ he is allocated and the amount of money
$\pi_i$ he is charged for it. We will represent it by a function $u_i(\theta_i,
x_i, \pi_i)$. Moreover, we will consider a set $P \subseteq \R_+^n$ that
specifies the set of feasible allocations. We will call such set the
{\em environment}.

\columnsversion{ \begin{defn} }{ \begin{Definition}}
An auction for this setting
consists of two mappings: the allocation
$x:\times_i \Theta_i \rightarrow P$ and the payment $\pi:\times_i \Theta_i
\rightarrow \R_+^n$. The auction is said to be individually rational if
$u_i(\theta_i, x_i(\theta), \pi_i(\theta)) \geq 0$ for all type vectors $\theta
= (\theta_1, \hdots, \theta_n)$. The auction is said to be incentive-compatible
(a.k.a. truthful) if no player can improve his utility by misreporting his
type, i.e.\columnsversion{:}{, for all $\theta'_i, \theta_i \in \Theta_i$ and
all $i \in [n]$:}
\columnsversion{$$ u_i(\theta_i, x_i(\theta_i, \theta_{-i}), \pi_i(\theta_i,
\theta_{-i})) \geq
u_i(\theta_i, x_i(\theta'_i, \theta_{-i}), \pi_i(\theta'_i, \theta_{-i})),
\forall \theta'_i, \theta_i \in \Theta_i, \forall i $$ }{$$ u_i(\theta_i,
x_i(\theta), \pi_i(\theta)) \geq
u_i(\theta_i, x_i(\theta'_i, \theta_{-i}), \pi_i(\theta'_i, \theta_{-i})) $$}
\columnsversion{ \end{defn} }{ \end{Definition}}
In this paper, we will be particularly interested in agents with
budget-constrained utility functions. For this setting $\Theta_i = \R_+$
representing the value $v_i$ the agent has for one unit of the good. There is a
{\em public} budget $B_i$ such that: $u_i(v_i, x_i, \pi_i) = v_i \cdot x_i -
\pi_i$ if $\pi_i\leq B_i$ and $-\infty$ otherwise. Incentive compatibility for
this setting means that the agents don't have incentives to misreport their
value. For this setting, we are interested in auctions producing Pareto-optimal
outcomes:

\columnsversion{ \begin{defn} }{ \begin{Definition}}
 Given an auction with environment $P$ and agents equipped with
budget-constrained utility functions, $B_i$ being the public budgets, we say
that an outcome $(x,\pi), x \in P, \pi \leq B$ is {\em Pareto optimal} if there
is no alternative outcome $x' \in P, \pi' \leq B$ such that $v_i x'_i - \pi'_i
\geq v_i x_i - \pi_i$ for all $i$, $\sum_i \pi'_i \geq \sum_i \pi_i$ and at
least one of those inequalities is strict. 
\columnsversion{ \end{defn} }{ \end{Definition}}

For the remainder of this paper, we will assume that budgets are public for
simplicity. For multi-unit auctions (i.e. $P = \{x; \sum_i x_i \leq s \}$), our
results extend to private budgets by applying the budget elicitation schemes in
Bhattacharya el al \cite{Bhattacharya10}.
\section{Online Supply Model}\label{sec:online_supply_model}

We consider auctions where the feasibility set is not known in advance to the
auctioneer. For each time $t \in \{0,\hdots, T\}$, we associate an environment
$P_t \subseteq \R_+^n$, which keeps track of the allocations done in times $t'
= 1..t$ \footnote{We want to stress the fact that $P_t$ {\em doesn't} represent
the set of allocations allowed {\em at} time $t$, but the set of allocations
allowed {\em until} time $t$. The set of new possible allocations in time $t$
is the difference between $P_t$ and $P_{t-1}$}.
  In each step, the mechanism needs to output an
allocation vector $x^t = (x_1^t, \hdots, x_n^t) \in P_t$ and a payment
vector $\pi^t = (\pi_1^t, \hdots, \pi_n^t) \geq 0$ by augmenting $x^{t-1}$ and
$\pi^{t-1}$. Given a set of desirable
properties, we would like to maintain them for all $t$. To make the problem
tractable, have to restrict the set of possible histories
$\{P_t\}_{t\geq 0}$. We do so by defining a partial ordering $\preccurlyeq$
on the set of feasibility constraints such that  if $t \leq s$ then $P_t
\preccurlyeq P_s$.

Our main goal is to design auctions where the auctioneer can allocate and charge
payments {\em `on the fly'}. The auctioneer will face a set of
environments $P_1 \preccurlyeq P_2 \preccurlyeq \hdots \preccurlyeq P_t$ and at
time $t$, he needs to allocate $x^t \in P_t$ and charge $\pi^t$, maintaining a
set of desirable properties. He doesn't know if $P_t$ will be the final outcome,
or if some new environment $P_{t+1} 
\succcurlyeq P_t$ will arrive, in which case he will need to augment $x^t \in
P_t$ to an allocation $x^{t+1}$ in $P_{t+1}$. It is crucial that his decision
at time $t$ doesn't depend the knowledge about $P_{t+1}$.

\columnsversion{ \begin{defn}[Online Supply Model] }{ \begin{Definition}[Online
Supply Model] }
Consider  a family of feasibility allocation constraints
indexed by $\F$, i.e, for each $f \in \F$ associate a set of feasible
allocation vectors $P^f \subseteq \R^n_+$ (a set $P^f$ is often called
environment). Also, consider a partial order $\preccurlyeq$ defined over $\F$
such that if $f \preccurlyeq f'$ then $P^f \subseteq P^{f'}$.
An auction for environment $P^f$ consists of functions $x^f: \Theta =
\times_i \Theta_i \rightarrow P^f$ and $\pi^f : \Theta \rightarrow \R^n_+$. 

An \textbf{auction in the strong online supply model} for 
 $(\F, \preccurlyeq)$ is a family of auctions such that $x^f
\leq x^{f'}$ and $\pi^f \leq \pi^{f'}$ whenever $f \preccurlyeq f'$.
Moreover, we say that the auction satisfies a certain property if it satisfies
this property for each $f$ (e.g. the auction is incentive compatible if for
each $f \in \F$, $(x^f, \pi^f)$ is an incentive compatible auction).

An \textbf{auction in the weak online supply model} for $(\F, \preccurlyeq)$ is
essentially the same, except that we drop the requirement of $\pi^f \leq
\pi^{f'}$. The intuition is that we are required to allocate goods online, but
are allowed to charge payments only in the end.
\columnsversion{ \end{defn} }{ \end{Definition}}

The main idea behind the definition is that if at some point the auctioneer
runs the auction $(x^f, \pi^f)$ for some environment $P^f$ and at a later time
some more goods arrive perhaps with new constraints such that the environment is
augmented to $P^{f'}$ with $f' \succcurlyeq f$, then the auctioneer can run
$(x^{f'}, \pi^{f'})$ and augment the allocation of player $i$ by $x_i^{f'}(v) -
x_i^f(v)$ goods and charge him more $\pi_i^{f'}(v) - \pi_i^f(v)$. 

\begin{example}[Multi-unit auctions]
 Let $\Delta_s = \{x \in \R^n_+; \sum_i
x_i \leq s\}$ and define
$\F^{\text{MU}} = \{\Delta_s; s \geq 0\}$ and let $\Delta_s
\preccurlyeq^{\text{MU}} \Delta_t$ iff $s \leq t$. Let the value of player $i$ for one unit of the good, $v_i$, lies in $\Theta_i = \R_+$. Now $u_i =
v_i x_i - \pi_i$. Thus, we are in a simple multi-unit auction setting. In this
setting, VCG is incentive compatible, individually rational and efficient (in
the sense that it has those three properties once run for each $\Delta_s$)
auction in the strong online model for $(\F^{\text{MU}}
, \preccurlyeq^{\text{MU}})$.
\end{example}

\columnsversion{
\begin{example}[Multi-unit auctions with 
capacities]\label{example:capacities}}{
\begin{example}[Multi-unit auctions with 
capacities]\label{example:capacities} }
Curiously, if players have capacity constraints, i.e., their utilities are $u_i
= v_i \min\{x_i, C_i\} - \pi_i$, then the VCG allocations for
$(\F^{\text{MU}} , \preccurlyeq^{\text{MU}})$ are still
monotone in the supply, but the payments are not. For example, consider two
agents with values $v_1 = 1, v_2 = 2$ and capacities $C_1 = C_2 = 1$. With
supply $1$, one item is allocated to player $2$ and he is charged $1$. With
supply $2$, both players get one unit of the item, but the VCG prices are zero.
Therefore, there is no incentive compatible, individually rational and
efficient in the strong online model. Babaioff, Blumrosen and Roth \cite{
babaioff10} strengthen this result showing that no $\Omega(\log \log
n)$-approximately efficient auction exists in the strong online model.
\end{example}

\begin{example}[Polymatroidal auctions]
Now, let $\F^{\text{PM}}$ be the set of all polymatroidal domains and consider
the naive-partial-order  $\preccurlyeq^{\text{N}}$ to be such that $f
\preccurlyeq f'$ iff $P^{f} \subseteq P^{f'}$. The VCG is not even online in
the weak sense for $(\F^{\text{PM}}, \preccurlyeq^{\text{N}})$.
Consider the following example:
\columnsversion{
$$P^{f} = \{x \in \R^2_+; x_1 \leq 2, x_2 \leq 2, x_1 + x_2 \leq 3\} 
\quad \text{   and   } \quad
P^{f'}
= \{x \in \R^2_+; x_1 + x_2 \leq 4\}$$}{$$P^{f} = \{x \in \R^2_+; x_1 \leq 2,
x_2 \leq 2, x_1 + x_2 \leq 3\}$$
$$P^{f'} = \{x \in \R^2_+; x_1 + x_2 \leq 4\}$$}
then clearly $f \preccurlyeq f'$ but if $v_1 > v_2$. $x^{f}(v) = (2,1)$ but
$x^{f'} = (4,0)$ violating the monotonicity property. But now, let's define a
different partial order $\preccurlyeq^{\text{PM}}$ such that
$f \preccurlyeq^{\text{PM}} f'$ if there is a polymatroid $P'$ such that
$P^{f'} = P^{f} + P'$ where the sum is the Minkowski sum. In Lemma
\ref{lemma:polymatroidal_online_supply} (Appendix \ref{appendix:missing_proofs})
we show that VCG
is an auction in the strong online model for
$(\F^{\text{PM}}, \preccurlyeq^{\text{PM}})$.
\end{example}

One interesting property of incentive-compatible auctions in the online supply
model is that utilities are monotone with the supply. If bidders have the option
of leaving in each timestep collecting their current allocations for their
current payment, they still (weakly) prefer to stay until the end of the
auction.

\begin{lemma}[Utility monotonicity]\label{lemma:utility-monotonicity}
Consider a setting where agents have single-parameter valuations $\Theta_i =
\R_+$ and quasilinear utilities $u_i = v_i x_i - \pi_i$.
Given a truthful auction in the weak online supply model and $f \preccurlyeq
f'$, then the utility of agent $i$ increased with the supply, i.e.:
$u^{f'} = v_i x^{f'} - \pi^{f'}_i \geq  v_i x^{f} - \pi^{f}_i = u^{f} $
\end{lemma}

\begin{proof}
 The proof follows directly from Myerson's characterization \cite{myerson-81} of
payments in quasi-linear settings: \columnsversion{$u^{f'} = v_i x^{f'}_i -
\pi^{f'}_i =
\int_0^{v_i} x^{f'}_i(u) du \geq
\int_0^{v_i} x^{f}(u) du = v_i x^{f} - \pi^{f}_i = u^f $.}{$$
\begin{aligned}
u^{f'} & = v_i x^{f'}_i - \pi^{f'}_i = \int_0^{v_i} x^{f'}_i(u) du 
\geq \int_0^{v_i} x^{f}(u) du = \\ & = v_i x^{f} - \pi^{f}_i = u^f
\end{aligned}$$}
\end{proof}

\section{Clinching Auctions and Supply
Monotonicity}\label{sec:multi-unit-auctions}

Our main theorem states that the Adaptive Clinching Auction (defined in
Dobzinski, Lavi and Nisan \cite{dobzinski12} and Bhattacharya el al
\cite{Bhattacharya10}) is an incentive-compatible auction in the strong online
supply model for budget constrained agents in the multi-unit setting. Formally:

\begin{theorem}\label{thm:main_theorem_multi_units}
 Given $n$ agents with public budgets $B_i$ and single-dimensional types $v_i
\in \R_+$ such that their utility is given by $u_i = v_i x_i - \pi_i$ if $\pi_i
\leq B_i$ and $u_i = -\infty$ otherwise, the Adaptive Clinching Auction is
an auction in the strong online supply model for $(\F^{\text{MU}} ,
\preccurlyeq^{\text{MU}})$. In other words, if $x(v,B,s)$ and $\pi(v,B,s)$ is
the outcome of the auction for valuation profile $v$, budgets $B$ and supply
$s$, then if $s \leq s'$, then: $x(v,B,s) \leq x(v,B,s')$ and $\pi(v,B,s) \leq
\pi(v,B,s')$.
\end{theorem}

Notice that this is in sharp contrast with what happens in Example
\ref{example:capacities} where getting a Pareto optimal auction in the strong
online model is not possible, not even in an approximate way \footnote{for that
setting, since there are no budgets, Pareto optimality boils down to
efficiency.}. This is somewhat surprising, since capacity constraints on the
allocations are usually more nicely-behaved compared to budget constraints.

Before proving the result, we review the Adaptive Clinching Auction
\cite{Bhattacharya10, dobzinski12},
presenting it in a way which will be more convenient for the proof.

\subsection{Adaptive Clinching Auction}

The clinching auction takes as input the valuation profile $v$, the budget
profile $B$ and the initial supply $s$, then it runs a procedure based on the
ascending price framework to determine final allocation and payments. There is a
price clock $p$, and for each price, the auction mantains\footnote{note
that here we prefer to index the ascending process by the price itself rather
then an external variable, like in Bhattacharya el al \cite{Bhattacharya10}.}
$x_i(p)$ denoting the current allocation of player $i$ and $B_i(p)$, which is
the current remaining budget of player $i$. Initially, $x_i(0) = 0$ and $B_i(0)
= B_i$, their initial budget. For each $p$, the auction
defines the values of the right-derivatives $\partial_p x_i(p)$ and $\partial_p
B_i(p)$ and described its behavior in the points in which it is discontinuous.
Notice we will use $\partial_p f(p)$ to denote the right-derivative of $f$ at
$p$.

For simplicity, we define the auction and prove our results for valuation
profiles $v$ such that $v_i \neq v_j$ for each $i \neq j$ (we call it a profile
in generic form) This is mainly a technical assumption to avoid
over-complicating the statement and the proof. 
In Appendix \ref{appendix:repeated-values}, we extend this for any valuation
profile $v$.
For the definition and its subsequent discussion, we will use the
following implicitly defined notation:\\

\begin{compactitem}
 \item remnant supply: $S(p) = s - \sum_i x_i(p)$
 \item active players: $A(p) = \{i; v_i > p\}$
 \item clinching players: \columnsversion{$C(p) = \{i \in A(p); S(p) = \sum_{j
\in A(p)
\setminus i} \frac{B_j(p)}{p} \}$}{$$C(p) = \{i \in A(p); S(p) =
\textstyle\sum_{j
\in A(p)
\setminus i} \frac{B_j(p)}{p} \}$$}
 \item maximum remaining budget:\columnsversion{$B_*(p) = \max_{i \in A(p)}
B_i(p)$}{$$B_*(p) = \max_{i \in A(p)} B_i(p)$$}
 \item for any function $f$, let $f(\bar{p} -) = \lim_{p \uparrow \bar{p}} f(p)$
and $f(\bar{p} +) = \lim_{p \downarrow \bar{p}} f(p)$
\end{compactitem}

\begin{defn}[Adaptive Clinching Auction]
\label{defn:adaptive_clinching_auction}
 Given as input a valuation vector $v$ in generic form, a budget vector $B$ and
initial supply $s$, consider the functions $x_i(p), B_i(p)$ such that:

\begin{enumerate}
 \item[(i)] $x_i(0) = 0$ and $B_i(0) = B_i$.
 \item[(ii)] $\partial_p x_i(p) = \frac{S(p)}{p}$ and $\partial_p B_i(p) =
-S(p)$ if
$i \in C(p)$ and  $\partial_p x_i(p) = \partial_p B_i(p) = 0$ otherwise.
 \item[(iii)] the functions $x_i$ and $B_i$ are right-continuous at all points
$p$, i.e., $x_i(p) = x_i(p+)$ and $B_i(p) = B_i(p+)$ for all $p$ and it is
left-continuous at all points $p \notin \{v_1, \hdots, v_n\}$, i.e., $x_i(p-) =
x_i(p)$ and $B_i(p-) = B_i(p)$ for all $p \notin \{v_1, \hdots, v_n\}$
 \item[(iv)] for $p = v_i$, let \columnsversion{$\delta_j = \left[ S(v_i -) -
\sum_{k \in
A(v_i) \setminus j} \frac{B_k(v_i -)}{v_i} \right]^+$.}{$$\delta_j = \left[
S(v_i -) - \textstyle\sum_{k \in A(v_i) \setminus j} \frac{B_k(v_i -)}{v_i}
\right]^+$$} For $j \in A(v_i)$, let $x_j(v_i) =
x_j(v_i -) + \delta_j$ and $B_j(v_i) = B_j(v_i -) - v_i \delta_j$ and for $j
\notin A(v_i)$, $x_j(v_i) = x_j(v_i-)$ and $B_j(v_i) = B_j(v_i -)$.
\end{enumerate}

The existence and uniqueness of those functions follow from elementary 
real analysis. The outcome associated with $v,B,s$ is $x_i = \lim_{p \rightarrow
\infty} x_i(p)$ and $\pi_i = B_i(0) - \lim_{p \rightarrow \infty} B_i(p)$.
Notice that this is well defined since $x$ and $B$ are constant for $p > \max_i
v_i$.
\end{defn}

The verb {\em clinch} means acquiring goods that are underdemanded at the
current price. So clinching a $\delta_i$ amount at price $p$ means receiving
$\delta_i$ amount of the good and paying $\delta_i p$ for it. When we refer to a
player clinching some amount, either we refer to the infinitesimal clinching
happening in (ii) or the player clinching positive units in (iv).

The reader familiar with the clinching auction for indivisible goods will
notice that the definition above is nothing more than the limit as $\epsilon
\rightarrow 0$ of this auction
run with $\frac{1}{\epsilon} s$ indivisible goods and valuations $\epsilon v_i$
per unit. This auction satisfies all the desirable properties for
multi-unit auctions with budgets:

\begin{theorem}[Bhattacharya et al \cite{Bhattacharya10}] The Adaptive
Clinching Auction in Definition \ref{defn:adaptive_clinching_auction} is
incentive-compatible, individually-rational, budget-feasible and produces
Pareto-optimal outcomes.
\end{theorem}

As one can possibly guess, it is possible to solve the differential equation in
each interval between two adjacent values of $v_i$ and give an explicit
description of the clinching auction. We do that in Appendix
\ref{appendix:explicit}. Nevertheless,
we mostly prove our results using the differential form in Definition
\ref{defn:adaptive_clinching_auction} which is more insightful than the
explicit version.

\begin{example}\label{example:clinching_run}
 At this point, it is instructive to consider an example of the auction.
Consider an auction between $n = 4$ players with valuations $v= [9,10,11,5.7]$
and $B = [3,2,1,.5]$. The functions $x_i(p), B_i(p)$ are depicted in Figure
\ref{fig:run_auction}. For $p < p_0^1 = 3.5$, the clinching set $C(p)$ is
empty. At this price $S(p_0^1) = 1 = \frac{2+1+0.5}{3.5} = \sum_{j \neq 1}
\frac{B_j}{p}$, so player $1$ alone begins clinching.

Since he is clinching alone for a while, $x_1(p) = s - S(p)$. Now by derivating
this expression we get that $\frac{S(p)}{p} = \partial_p x_1(p) = - \partial_p
S(p)$. Solving for the supply with the condition that $S(p_0^1) = 1$, we get:
$S(p) = \frac{s p_0^1}{p}$ and $x_1(p) = s - \frac{s p_0^1}{p}$. This continues
while no other player enters the clinching set. The supply function $S(p)$ is
illustrated in the first part of Figure \ref{fig:supply_and_wishful}.

Notice that for this period, the budget of $1$ is being spent while the budgets
of the other agents are intact. Eventually, the budget of $1$ meets the budget
of $2$, and at this point, those two players are indistinguishable from the
perspective of the differential procedure as long as both are active.
Therefore, both start spending their budget at the same rate and acquiring
goods at the same rate $\partial_p x_1(p) = \partial_p x_2(p) = S(p)$. Since
from this point on $S(p) = s - x_1(p) - x_2(p)$, then $\frac{S(p)}{p} =
\partial_p x_1(p) = - \frac{1}{2} \partial_p S(p)$. Solving again for the
supply and the boundary condition $S(p_0^2) = \frac{s p_0^1}{p_0^2}$ we get:
$S(p) = \frac{s p_0^1 p_0^2}{p^2}$. Using this, one can  calculate
$x_1(p)$ and $x_2(p)$. Both continue clinching at the same rate until the price reaches
$p = v_4$, where player $4$ exits the active set prompting the agents to
clinch a positive amount according to (iv).

Their allocation $x_i(p)$ and budgets $B_i(p)$ are discontinuous at this point,
but continue to follow the differential procedure after this point, having
their budgets all equal (not coincidentally, as we will see  in Lemma
\ref{lemma:budget_profile_format_pre}), until price $p = v_1$ is reached and
player $1$ exits the active set. At this point, all the remaining active players
 clinch a positive amount according to (iv) that exhausts the supply.
Therefore, allocations and budgets are constant from this point on.
\end{example}

\columnsversion{\begin{figure}}{\begin{figure*}}
\centering
\includegraphics{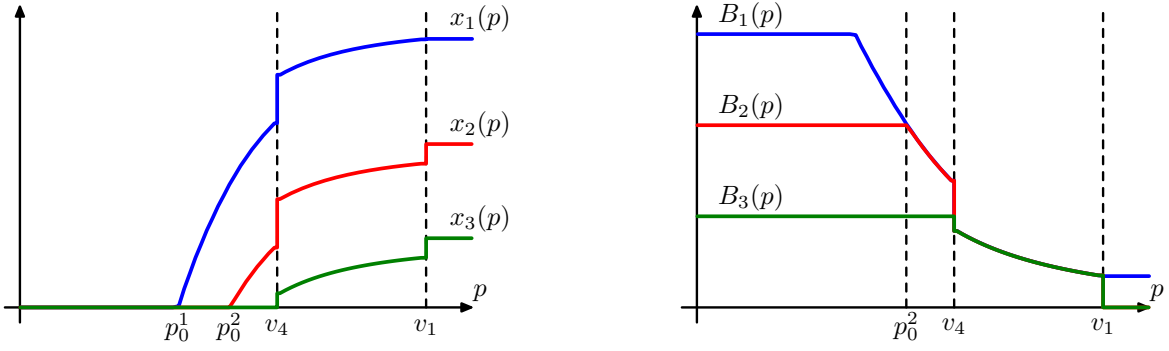}
\caption{The functions $x(p)$ and $B(p)$ for an auction in Example
\ref{example:clinching_run}}
\label{fig:run_auction}
\columnsversion{\end{figure}}{\end{figure*}}

\columnsversion{\begin{figure}}{\begin{figure*}}
\centering
\includegraphics{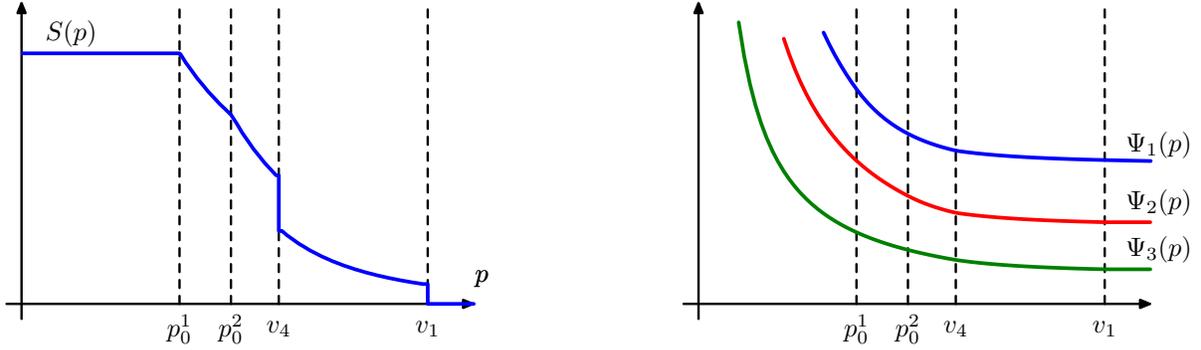}
\caption{Supply $S(p)$ and wishful allocation $\Psi(p)$ for an auction in
Example
\ref{example:clinching_run}}
\label{fig:supply_and_wishful}
\columnsversion{\end{figure}}{\end{figure*}}

One important tool in analyzing this auction is the concept of the {\em wishful
allocation}. We define a $\Psi_i(p)$ as a function of $x_i(p)$ and $B_i(p)$
which is continuous even at the points where $x_i(p)$ and $B_i(p)$ are not. It
is carefully set up so that the discontinuities from both functions cancel out.
Intuitively, it represents a sum of what the player acquired already at the
current price $x_i(p)$ with the maximum amount he would like to acquire at this
price, which is $\frac{B_i(p)}{p}$.

\begin{defn}[Wishful allocation]
The wishful allocation is defined as $\Psi_i(p) = x_i(p) + \frac{B_i(p)}{p}$.
\end{defn}

\begin{lemma}\label{lemma:wishful_allocation_properties}
The wishful allocation is continuous and right-differentiable for all $p \geq
0$. Moreover, its  right-derivative is given by: $\partial_p \Psi_i(p) =
-\frac{B_i(p)}{p^2}$.
\end{lemma}

The proof is elementary and can be found in Appendix
\ref{appendix:missing_proofs}. Since $\Psi_i(p) \geq x_i(p)$ and is a monontone
non-increasing function converging to the final allocation as $p \rightarrow
\infty$, it constantly gives us an upper bound of the final allocation.

Now, we study some other properties of the above auction, which will be useful
in the proof of our main theorem. First we prove a Meta Lemma
that sets the basic structure for most of our proofs. The lemma is based on
elementary facts of real analysis. Its proof, as well as other missing
proofs of this section can be found in Appendix \ref{appendix:missing_proofs}.

\begin{metalemma}\label{the_metalemma}
Given a property $\Lambda$ that depends on $p$, if we want to prove for all
$p \geq p_0$, it is enough to prove the following facts:
\begin{compactenum}
 \item[(a)] it holds for $p = p_0$.
 \item[(b)] if $\Lambda$ holds for $p$, then there is some $\epsilon_p > 0$
 such that $\Lambda$ holds for $[p,p+\epsilon_p)$
 \item[(c)] if $\Lambda$ holds for all $p'$ such that $p_0 \leq p' < p$, then
$\Lambda$ also holds for $p$.
\end{compactenum}
\end{metalemma}

For most properties $\Lambda$ that we want to prove about
the Adaptive Clinching
Auction, part (a) is easy to show, part (b) requires using the right-continuity of
the function and the value of the right-derivatives given in item (ii) of
Definition \ref{defn:adaptive_clinching_auction} and part (c) is usually
proved using continuity for $p \notin \{v_1, \hdots, v_n\}$ and using part (iv)
of Definition \ref{defn:adaptive_clinching_auction}.\\

The first two lemmas (whose proof is based on the Meta-Lemma) state that once a
player start acquiring goods (i.e. $\partial_p x_i(p) > 0$), he continues to do
so for all the prices until $p$ becomes equal to his value $v_i$.
All the proofs can be found in Appendix \ref{appendix:missing_proofs}.

\begin{lemma}\label{lemma:once_clinching_always_clinching}
Once a player $i$ enters the clinching set, then he is in the clinching set
until he becomes inactive, i.e., if $i \in C(p)$ for some $p$, then $i \in
C(p')$ for all $p' \in [p,v_i)$.
\end{lemma}

\begin{lemma}\label{lemma:supply_inequality}
For each price $p$ and each active player $i$, $S(p) \leq \sum_{j \in A(p)
\setminus i} \frac{B_j(p)}{p}$.
\end{lemma}

\begin{cor}\label{cor:clinching_set_after_suddenly_clinching}
If at price $p = v_j$, player $i \in A(v_j)$ acquires any positive amount of the
good $\delta_i > 0$, then he enters in the clinching set (if he wasn't
previously), i.e., $i \in C(v_j)$.
\end{cor}

A crucial observation for our proof is that the evolution of the profile of
remaining budgets follows a very structured format. At any given price, 
the remaining budget of an agent is either his original budget or the
maximum budget among all agents. It is instructive to observe that in  Figure
\ref{fig:run_auction}.

\begin{lemma}\label{lemma:budget_profile_format_pre}
For each price $p$, if $C(p) \neq \emptyset$, then $C(p) = \{ i\in A(p); B_i(p)
= B_*(p)\}$. 
\end{lemma}

\begin{cor}\label{cor:budget_profile_format}
 For each $i \in A(p)$, $B_i(p) = \min\{B_i(0), B_*(p)\}$.
\end{cor}

\subsection{Supply Monotonicity}

Now we are ready to prove  Theorem \ref{thm:main_theorem_multi_units} which
is our main result. For that we fix a budget profile $B$ and a valuation
profile $v$ in generic form (i.e. $v_i \neq v_j$ for $i \neq j$, which is not
needed for the proof and is mainly intended to simplify the exposition; See
Appendix~\ref{appendix:repeated-values}).  
Now, we consider two executions of the adaptive clinching auction. One with
initial supply $s^\base$ which we call the {\em base auction} and one with
initial supply $s^\aug \geq
s^\base$ which we call the {\em augmented auction}. Running the base and
augmented auction with the same valuations and budgets we get functions
$x^\base (p), B^\base (p)$ and $x^\aug (p), B^\aug (p)$. From this point on, we
use superscripts $\base$ and $\aug$ to refer to the base and augmented auctions
respectively. For the set of active players at a given price, we omit the
superscript, since $A^\base(p) = A^\aug(p)$ for all $p$. 

As the first step toward the proof of Theorem
\ref{thm:main_theorem_multi_units}, we prove
that the payments are monotone with the supply, that the final payment of each
agent in the augmented auction is higher than in the base auction:

\begin{prop}[Payment Monotonicity]\label{prop:payment_monotonicity}
 Given the base and augmented auction as defined above, then for all $p \geq 0$
and all agents $i$, $B_i^\base(p) \geq B_i^\aug(p)$.
\end{prop}

The full proof of this proposition is delicate and involves keeping track of
many invariants as the allocation and budget profiles evolve with prices. We
delay the proof to Appendix \ref{appendix:proof_budget_monotonicity_prop}. Here
we present a warm-up to the proof that deals with the case where all values
$v_i$ are very large.  This special case will highlight the core of the proof
which is essentially a coupling argument. Also, we keep the discussion here
informal and delay the formal arguments to Appendix
\ref{appendix:proof_budget_monotonicity_prop}.\columnsversion{}{\\}

\begin{warmup}
 If valuations are large compared to budgets, $C^\aug(p) = C^\base(p)
= [n]$ for some $p < \min_i v_i$ and once $\bar{p} = \min_i v_i$ is reached,
some player $i$ becomes inactive and all the other players $j \neq i$, clinch
their entire demand $\delta_j = \frac{B_j(\bar{p} -)}{\bar{p}}$.

This case is nice because it allows us to ignore part (iv) of Definition
\ref{defn:adaptive_clinching_auction} and simply analyze the continuous
function defined in $[0,\bar{p})$ by part (ii) of the definition. We start by
defining $p_0^\aug = \min\{p; C^\aug(p) \neq \emptyset\}$ and  $p_0^\base =
\min\{p; C^\base(p) \neq \emptyset\}$. Since the supply is larger in the
augmented auction, $p_0^\aug < p_0^\base$. Then we can divide the analysis in
three intervals: in the interval $[0, p_0^\aug)$ where no player is clinching,
so budgets are constant, i.e., equal to the initial budget. In the interval
$[p_0^\aug, p_0^\base)$ no player is clinching in the base auction but some are
clinching in the augmented auction, so clearly the remaining budgets are larger
in the base auction. For the remaining interval $[p_0^\base, \infty)$ we can
use Corollary \ref{cor:budget_profile_format} to see that all we need to show
is that since both clinching sets are non-empty we just need to prove that
$B_*^\aug(p) \leq B_*^\base(p)$ for all $p \in [p_0^\base, \infty)$. This is
true for $p = p_0^\base$ by continuity of $B_*^\aug$ and $B_*^\base$ (since $p
< \min_i v_i$). Now we argue that if
$B_*^\aug(p) \leq B_*^\base(p)$ for some $p \geq p_0^\base$, then $B_*^\aug(p')
\leq B_*^\base(p')$ for $p' \in [p,p+\epsilon)$ for some $\epsilon > 0$. Then
we invoke the Meta-Lemma to extend this to all $p \geq p_0^\base$. Now, in
order to prove that we analyze two cases: if $B_*^\aug(p) < B_*^\base(p)$ then
by continuity of $B_*^\aug$ and $B_*^\base$, there exists some
$\epsilon > 0$ such that $B_*^\aug(p') < B_*^\base(p')$ for $p' \in
[p,p+\epsilon)$. Now, if $B_*^\aug(p) = B_*^\base(p)$, then by
Corollary \ref{cor:budget_profile_format}, $B_i^\aug(p) = B_i^\base(p)$ for all
agents $i$ and moreover, the remnant supply is the same $S^\aug(p) =
S^\base(p)$, since $S(p) = \sum_{i \in A(p)} \frac{B_i(p)}{p} -
\frac{B_*(p)}{p}$. Notice that the evolution of budgets in $p' \geq p$ depends
only on $B(p)$ and $S(p)$ and since those are equal for the augmented and for
the base auction, then $B^\aug_i(p') = B^\base_i(p')$ for all $p' \geq p$. We
say that at this point, the auctions get {\em fully coupled}. This completes
the discussion. The heart of the proof is to show that the maximum
budget of the base auction stays higher then the one in the augmented auction.
If eventually they meet, then the two auctions become {\em fully coupled}, in
the sense that they evolve in the same way from this price on.
\end{warmup}

Now, we want to establish allocation monotonicity, i.e., that $x_i^\aug(p) \geq
x_i^\base(p)$ for all $p \geq 0$. We will prove a stronger claim, that the {\em
wishful allocation} $\Psi_i$ is monotone in the supply, i.e., $\Psi_i^\aug(p)
\geq \Psi_i^\base(p)$ for all $p \geq 0$. 
\columnsversion{}{\\}

\begin{prop}[Allocation Monotonicity]\label{prop:allocation_monotonicity}
 For all $p \geq 0$
and all agents $i$, the following invariant holds: $\Psi^\base_i(p) \leq
\Psi_i^\aug(p)$.
\end{prop}

\begin{proof}
 This proof follows from combining Proposition
\ref{prop:payment_monotonicity} and Lemma
\ref{lemma:wishful_allocation_properties}. For small values of $p$,
$\Psi^\base_i(p) \leq \Psi_i^\aug(p)$ is definitely true, since both are equal
to $\frac{B_i}{p}$. Now, if it is true for some small $p$, then it is true for
any $p' \geq p$, since:
\columnsversion{$$ \Psi^\aug_i(p') = \Psi^\aug_i(p) - \int_p^{p'}
\frac{B_i^\aug(\rho)}{\rho^2}
d\rho \geq \Psi^\base_i(p) - \int_p^{p'} \frac{B_i^\base(\rho)}{\rho^2} d\rho =
\Psi^\base_i(p').$$}{$$\begin{aligned} \Psi^\aug_i(p') & = \Psi^\aug_i(p) -
\int_p^{p'}
\frac{B_i^\aug(\rho)}{\rho^2}
d\rho \geq \\ & \geq \Psi^\base_i(p) - \int_p^{p'}
\frac{B_i^\base(\rho)}{\rho^2} d\rho = \Psi^\base_i(p'). \end{aligned}$$}
\end{proof}

The proof of our main theorem follows immediately from Propositions
\ref{prop:payment_monotonicity} and
\ref{prop:allocation_monotonicity}.\columnsversion{}{\\}

\begin{proofof}{Theorem \ref{thm:main_theorem_multi_units}}
For the allocation monotonicity, Proposition \ref{prop:allocation_monotonicity}
implies that
$x_i^\aug(p) + \frac{B_i^\aug(p)}{p} \geq x_i^\base(p) +
\frac{B_i^\base(p)}{p}$. Since $B_i^\aug(p) \leq B_i^\base(p)$, then clearly:
$x_i^\aug(p) \geq x_i^\base(p)$, taking $p \rightarrow \infty$ we get that for
each player $i$, the final allocation in the augmented auction and in the base
auction are such that $x_i^\aug \geq x_i^\base$.

The monotonicity of the payment function follows directly from Proposition
\ref{prop:payment_monotonicity}. The remaining budget in the end is larger in
the base auction then in the augmented auction for each agent. So, the final
payments are such that $\pi_i^\aug \geq \pi_i^\base$.

\end{proofof}

\section{Qualitative Description of the Adaptive Clinching
Auction}\label{sec:qualitative}

Once we established that the Adaptive Clinching Auction is an auction in the
online supply model, we have a feasible online allocation and pricing rule in
our hands,
i.e., a rule that tells how to allocate and charge for an $\epsilon$ amount of
the good when it arrives after $s$ supply has already been allocated. At this
point, it is worth studying qualitative
behavior of such an allocation rule, and develop more insights
that can 
guide us in the design of heuristic to apply in real-world ad auctions.

We do this analysis in details for $2$ players in Appendix
\ref{appendix:closed_form_2}. We observe that for the first units of supply,
the clinching auction behaves like VCG: allocating to the highest bidder and
charging the second highest bid. After that, depending on the relation between
$v_1$ and $v_2$, two distinct behaviors can happen: either at a certain point
the high value player gets his budget depleted, and the auction starts
allocating new arriving units to the low valued player (still charging for
those items) and then at a certain
further point, it starts splitting the goods among them (charging only the player
with non-depleted budget). Or alternatively, the
auction continues to allocate to the high-valued player but start charging him
a discounted version of VCG. Then when his budget gets depleted, the auction
starts splitting goods among the players (charging only the player with
non-depleted budget). We refer to Appendix \ref{appendix:closed_form_2} for a
detailed discussion of the intuition behind this online rule.

\bibliographystyle{abbrv}
\bibliography{sigproc}  
%
%
\appendix

\section{Missing Proofs in Sections \ref{sec:online_supply_model} and
\ref{sec:multi-unit-auctions} }\label{appendix:missing_proofs}

\begin{lemma}\label{lemma:polymatroidal_online_supply}
 VCG is an auction in the strong online model for
$(\F^{\text{PM}}, \preccurlyeq^{\text{PM}})$.
\end{lemma}

\begin{proof}
Assume that $P^{f}, P^{f'}, P'$ are defined respectively
by the monotone submodular functions $f,f',g$. If $P^{f'} = P^f + P'$, then by
McDiarmid's Theorem \cite{mcdiarmid75}, $f' = f+g$.

Now, let's remind how VCG allocated for this setting. If the polymatroid is
defined by $f$, VCG begins by sorting the players by their value (and breaking
ties lexicographically). So, we can assume $v_1 \geq \hdots \geq v_n$. Then it
chooses the outcome:
$$x_i = f([i])-f([i-1])$$
\columnsversion{$$\pi_i = v_{i+1} \cdot (f([i+1]\setminus i) -f([i-1]) -
x_{i+1}) +\sum_{j>i+1} v_j \cdot (f([j]\setminus i) -f([j-1]\setminus i) - x_j)
$$}{$$\begin{aligned}\pi_i & = v_{i+1} \cdot (f([i+1]\setminus i) -f([i-1]) -
x_{i+1}) + \\ & \qquad + \sum_{j>i+1} v_j \cdot (f([j]\setminus i)
-f([j-1]\setminus i) - x_j)
\end{aligned}$$}
where $[i]$ is an abbreviation for $\{1, \hdots, i\}$. Now, once we do this
for $f,f'$ we notice that the allocation and payments for $f'$ are simply the
sum of the allocation and payments for $f$ and $g$, hence they are monotone
along $\preccurlyeq^{\text{PM}}$.
\end{proof}

\begin{proofof}{Lemma \ref{lemma:wishful_allocation_properties}}
 The function $\Psi_i(p)$ is clearly continuous for $p \notin \{v_1, \hdots,
v_n\}$ and right-continuous everywhere. Now, we claim that it is also
left-continuous at $v_j$, i.e., $\Psi(v_j -) = \Psi(v_j)$. This fact is almost
immediate:
\columnsversion{ $$\Psi_i(v_j) = x_i(v_j) + \frac{B_i(v_j)}{v_j} = [x_i(v_j -) +
\delta_i] + \frac{[B_i(v_j-) - \delta_i v_j]}{v_j} = x_i(v_j -) + \frac{B_i(v_j
-)}{v_j} = \Psi_i(v_j -)$$}{$$\begin{aligned}\Psi_i(v_j) & = x_i(v_j) +
\frac{B_i(v_j)}{v_j} = \\ & = [x_i(v_j -) +
\delta_i] + \frac{[B_i(v_j-) - \delta_i v_j]}{v_j} =\\ & = x_i(v_j -) +
\frac{B_i(v_j
-)}{v_j} = \Psi_i(v_j -)\end{aligned}$$ }
Calculating its derivative is also easy:
\columnsversion{ $$\partial_p \Psi_i(p) = \partial_p \left[ x_i(p) +
\frac{B_i(p)}{p} \right] = \partial_p x_i(p) + \frac{\partial_p B_i(p)}{p} -
\frac{B_i(p)}{p^2} = - \frac{B_i(p)}{p^2}$$}{$$\begin{aligned}\partial_p
\Psi_i(p) & = \partial_p \left[ x_i(p) + \frac{B_i(p)}{p} \right] = \\ & =
\partial_p x_i(p) +
\frac{\partial_p B_i(p)}{p} - \frac{B_i(p)}{p^2} = - \frac{B_i(p)}{p^2}
\end{aligned}$$}
since $\partial_p x_i(p) = \frac{S(p)}{p} = -\frac{\partial_p B_i(p)}{p}$.
\end{proofof}

\begin{proofof}{ the Meta Lemma \ref{the_metalemma}}
 Let $F = \{p\geq p_0; \Lambda \text{ doesn't hold for } p\}$.  We want  to
show
that if the properties (a),(b),(c) in the statement hold, then $F = \emptyset$.
Assume for
contradiction that (a),(b),(c) hold but $F \neq \emptyset$. Let $\bar{p} = \inf
F$, i.e., the smallest $\bar{p}$ such that for all $\epsilon > 0$,
$[\bar{p},\bar{p}+\delta) \cap F \neq \emptyset$ for all $\delta > 0$.

Now, there are two possibilities:

(1) either $\bar{p} \notin F$, in this case
we can invoke (b) to see that there should be an $\epsilon > 0$ such that
$[\bar{p},\bar{p}+\epsilon) \cap F = \emptyset$ which contradicts the fact that
$\bar{p} = \inf F$.

(2) or $\bar{p} \in F$. By (a), we know $\bar{p} > p_0$. Then
we can use that by the definition of $\inf$, $\Lambda$ holds for all $p <
\bar{p}$, so we can invoke (c) to show that $\Lambda$ should hold for
$\bar{p}$. And again we arrive in a contradiction.
\end{proofof}

\begin{proofof}{Lemma \ref{lemma:once_clinching_always_clinching}}
 The proof is based on the Meta Lemma. Part (a) is trivial.

 For part (b), there is $\epsilon > 0$ such that in $[p,p+\epsilon)$ the active
set is the same as $A(p)$. We will show that if $i \in C(p)$, then $i \in
C(p')$ for all $p' \in [p, p+\epsilon)$, or in other words: $S(p') = \sum_{j
\in A(p) \setminus i} \frac{B_j}{p'}$. This equality holds for $p$. Now, we
will simply show that the derivative of both sides is the same in the
$[p,p+\epsilon)$ interval,
i.e.: $\partial_p  S(p) = \partial_p\sum_{j\in A(p) \setminus i}
\frac{B_j(p)}{p}$.

\columnsversion{$$ \begin{aligned}
\partial_p\sum_{j\in A(p) \setminus i} \frac{B_j(p)}{p} & =
\frac{p[\sum_{j \in A(p) \setminus i} \partial_p B_j(p)] - \sum_{j \in A(p)
\setminus i}  B_j(p) }{p^2} = \frac{1}{p} \sum_{j \in C(p) \setminus i}
\partial_p B_j(p) - \frac{1}{p} S(p) = \\
&  = - \sum_{j \in C(p) \setminus i}
\frac{1}{p}S(p) - \frac{1}{p} S(p) = -\sum_{j \in C(p)} \partial_p x_j(p) =
-\sum_{j \in A(p)} \partial_p x_j(p) = \partial_p S(p)
\end{aligned}$$}{
$$ \begin{aligned}
& \partial_p\sum_{j\in A(p) \setminus i}  \frac{B_j(p)}{p}  = \\ & \quad =
\frac{p[\sum_{j \in A(p) \setminus i} \partial_p B_j(p)] - \sum_{j \in A(p)
\setminus i}  B_j(p) }{p^2} = \\ & \quad = \frac{1}{p} \sum_{j \in C(p)
\setminus i}
\partial_p B_j(p) - \frac{1}{p} S(p) = \\
&  \quad = - \sum_{j \in C(p) \setminus i}
\frac{1}{p}S(p) - \frac{1}{p} S(p) \\ & \quad = -\sum_{j \in C(p)} \partial_p
x_j(p) = -\sum_{j \in A(p)} \partial_p x_j(p) = \partial_p S(p)
\end{aligned}$$}

 For part (c), it is trivial for $p \notin \{v_1, \hdots, v_n\}$ by
left-continuity: if $S(p') = \sum_{j \in A(p') \setminus i} \frac{B_j}{p'}$
for $p' < p$ and the functions involved are left-continuous, then it holds for
$p$. Now, for $p = v_j$, if $S(v_j -) = \sum_{k \in A(v_j -) \setminus i}
\frac{B_k(v_j -)}{v_j}$, then for $\delta_k$ as defined in (iii) of Definition
\ref{defn:adaptive_clinching_auction} we have:
\columnsversion{$$S(v_j)  = S(v_j -) - \sum_{k \in A(v_j)} \delta_k = \left[
\sum_{k \in A(v_j)\setminus i} \frac{B_k(v_j -) - \delta_k v_j}{v_j} \right] +
\frac{B_j(v_j -)}{v_j}- \delta_i = \sum_{k \in A(v_j)\setminus i} \frac{B_k(v_j
)}{v_j}$$}{$$\begin{aligned} S(v_j)  & = S(v_j -) - \sum_{k \in A(v_j)} \delta_k
= \\ & = \left[
\sum_{k \in A(v_j)\setminus i} \frac{B_k(v_j -) - \delta_k v_j}{v_j} \right] +
\frac{B_j(v_j -)}{v_j}- \delta_i = \\ & = \sum_{k \in A(v_j)\setminus i}
\frac{B_k(v_j)}{v_j} \end{aligned}$$}
since:
\columnsversion{$$\delta_i = \left[ S(v_j -) - \sum_{k \in A(v_j)\setminus i}
\frac{B_k(v_j
-)}{v_j} \right]^+ = \left[ \sum_{k \in A(v_j -)\setminus i} \frac{B_k(v_j
-)}{v_j} - \sum_{k \in A(v_j)\setminus i} \frac{B_k(v_j
-)}{v_j} \right]^+ = \frac{B_j(v_j -)}{v_j}$$}{
$$\begin{aligned} \delta_i & = \left[ S(v_j -) - \sum_{k \in A(v_j)\setminus i}
\frac{B_k(v_j
-)}{v_j} \right]^+ = \\ & = \left[ \sum_{k \in A(v_j -)\setminus i}
\frac{B_k(v_j
-)}{v_j} - \sum_{k \in A(v_j)\setminus i} \frac{B_k(v_j
-)}{v_j} \right]^+ = \\ & = \frac{B_j(v_j -)}{v_j} \end{aligned}$$}
\end{proofof}

\begin{proofof}{Lemma \ref{lemma:supply_inequality}}
 Again we prove it using the Meta Lemma. (a) is trivial, for (b) if $S(p) <
\sum_{j \in A(p) \setminus i} \frac{B_j(p)}{p}$, then by right-continuity the
strict inequality continues to hold in some region $[p,p+\epsilon)$. If $S(p) =
\sum_{j \in A(p) \setminus i} \frac{B_j(p)}{p}$ we can do the same analysis as
in Lemma \ref{lemma:once_clinching_always_clinching}. For (c) it is again
trivial for $p \notin \{v_1, \hdots, v_n\}$ by left-continuity and for $p =
v_j$ we use the fact that comes directly from the proof of the
previous lemma:
\columnsversion{\begin{equation}\label{eq:supply_update}
 S(v_j) - \sum_{k \in A(v_j)\setminus i} \frac{B_k(v_j
)}{v_j} = \left[ S(v_j -) - \sum_{k \in A(v_j -)\setminus i} \frac{B_k(v_j
-)}{v_j} \right] + \frac{B_j(v_j-)}{v_j} - \delta_i \leq 0
\end{equation}}{\begin{equation}\label{eq:supply_update}
\begin{aligned} 
& S(v_j) - \sum_{k \in A(v_j)\setminus i} \frac{B_k(v_j
)}{v_j} = \\ & \quad = \left[ S(v_j -) - \sum_{k \in A(v_j -)\setminus i}
\frac{B_k(v_j
-)}{v_j} \right] + \\ & \quad \quad + \frac{B_j(v_j-)}{v_j} - \delta_i \leq 0
\end{aligned}
\end{equation}}
by the definition of $\delta_i$, since:
\columnsversion{
$$\delta_i = \left[ S(v_j -) - \sum_{k \in A(v_j) \setminus i} \frac{B_k(v_j
-)}{v_j} \right]^+ \geq \left[ S(v_j -) - \sum_{k \in A(v_j -)\setminus i}
\frac{B_k(v_j
-)}{v_j} \right] + \frac{B_j(v_j-)}{v_j} $$}{
$$\begin{aligned} \delta_i & = \left[ S(v_j -) - \sum_{k \in A(v_j) \setminus i}
\frac{B_k(v_j
-)}{v_j} \right]^+ \geq \\ &  \geq \left[ S(v_j -) - \sum_{k \in A(v_j
-)\setminus i}
\frac{B_k(v_j
-)}{v_j} \right] + \frac{B_j(v_j-)}{v_j} \end{aligned}$$}
\end{proofof}

\begin{proofof}{Corollary \ref{cor:clinching_set_after_suddenly_clinching}}
 If $\delta_i > 0$, then, $\delta_i =  S(v_j -) - \sum_{k \in A(v_j)
\setminus i} \frac{B_k(v_j -)}{v_j} $. Substituting that in equation
(\ref{eq:supply_update}) we get that $S(v_j) = \sum_{k \in A(v_j)\setminus i}
\frac{B_k(v_j )}{v_j} $ and therefore $i \in C(v_j)$.
\end{proofof}

\begin{proofof}{Lemma \ref{lemma:budget_profile_format_pre}}
 It is easy to see that all bidders in the clinching set have the same
remaining budget, since if $i,i' \in C(p)$, then $\sum_{j \in A(p) \setminus i}
\frac{B_j(p)}{p} = S(p) = \sum_{j \in A(p) \setminus i'} \frac{B_j(p)}{p}$ and
therefore $B_i(p) = B_{i'}(p)$. Also, clearly, all players with the same budget
will be in the clinching set. The fact that the players clinching have the
largest budget follows directly from Lemma \ref{lemma:supply_inequality}.
\end{proofof}

\section{Proof of Proposition
\ref{prop:payment_monotonicity}}\label{appendix:proof_budget_monotonicity_prop}

\begin{proofof}{Proposition \ref{prop:payment_monotonicity}}

 The first part of the proof consists of showing that
clinching
starts first in the augmented auction. Then we divide the prices in three
intervals: in the first where no clinching happens in both auctions, in the
second where clinching happens only in the augmented auction and the third in
which clinching happens in both auctions. Then we prove the claim in each of
the intervals.\\

{\em First part of the proof:} Clinching starts earlier in the augmented auction

Let $p_0^\base = \min\{p; C^\base(p) \neq \emptyset\}$ and $p_0^\aug = \min\{p;
C^\aug(p) \neq \emptyset\}$. We claim that $p_0^\aug \leq p_0^\base$. In order
to see that, assume the contrary: $p^\base_0 < p_0^\aug$. At $p_0^\base$, there
is one agent $i$ such that $S^\base(p_0^\base) = \sum_{k \in
A(p_0^\base)\setminus i} \frac{B_k^\base(p_0^\base)}{p_0^\base}$. If $p_0^\base
\notin \{v_1, \hdots, v_n\}$, then by Corollary
\ref{cor:clinching_set_after_suddenly_clinching}, no budget was spent in neither
of the auctions at this price and no goods were acquired, so $S^\base(p_0^\base)
= s^\base$, $S^\aug(p_0^\base) = s^\aug$, $B_k^\base(p_0^\base) = B_k^\base(0)$
and $B_k^\aug(p_0^\base) = B_k^\aug(0)$. This implies that at this point
$S^\aug(p_0^\base) > S^\base(p_0^\base) = \sum_{k \in A(p_0^\base)\setminus
i} \frac{B_k^\aug(p_0^\base)}{p_0^\base}$, which contradicts Lemma
\ref{lemma:supply_inequality} for the augmented auction. Now, the case left to
analyze is the one where $p_0^\base = v_j$ for some $j \neq i$ and $i$ entered
the clinching set after acquiring a positive amount of good $\delta_i^\base > 0$
at price $v_j$. Then: $\delta_i^\base = s^\base - \sum_{k \in A(v_j)\setminus i}
\frac{B_k(0)}{v_j} > 0$. But in this case $\delta_i^\aug > 0$,
contradicting that $p_0^\aug > p_0^\base$.\\

{\em Second part  of the proof:} Proof for the first interval $[0,p^\aug_0)$.

For the $p$ in the interval  $[0, p^\aug_0)$, no clinching occurs, so
$B^\aug_i(p) = B^\aug_i(0) = B^\base_i(0) = B^\base_i(p)$.\\

{\em Third part of the proof:} Proof for the second interval $[p^\aug_0, p^\base_0)$.

In the interval $[p^\aug_0, p^\base_0)$, some players are acquiring goods in
the augmented auction but no player is neither acquiring goods nor paying
anything in the base auction, so: $B^\aug_i(p) \leq B^\aug_i(0) = B^\base_i(0) =
B^\base_i(p)$. \\

{\em Fourth part of the proof:} Proof for the third interval $[p^\base_0, \infty)$.

In this interval, both players are clinching. Now, we use the Meta Lemma to
show that for all $p \geq p^\base_0$, the property $B_i^\base(p) \geq
B_i^\aug(p)$ for all $i$ holds. 

For part (a) of the Meta Lemma, we need to show that $B_i^\base(p_0^\base) \geq
B_i^\aug(p_0^\base)$. If $p_0^\base \notin \{v_1,\hdots, v_n\}$ this follows
directly from continuity and the third part of the proof. If $p_0^\base = v_j$
for some $j$, then by the previous cases we know that $B^\aug_i(v_j -) \leq
B^\base_i(v_j -)$. We have that $B^\aug_i(v_j) = B^\aug_i(v_j -) - \delta_i^\aug
v_j$ and $B^\base_i(v_j) = B^\base_i(v_j -) - \delta_i^\base v_j$. Now we
analyze the clinched amounts $\delta_i^\aug$ and $\delta_i^\base$. If
$p_0^\aug = p_0^\base$, it is straightforward to see that $\delta_i^\aug \geq
\delta_i^\base$ and therefore $B^\aug_i(v_j) \leq B^\base_i(v_j)$. So, let's
focus on the case where $p_0^\aug < p_0^\base$. For this case:

\columnsversion{$$\begin{aligned} \delta_i^\aug & = \left[ S^\aug(v_j -) -
\sum_{k \in A(v_j)
\setminus i} \frac{B_k^\aug (v_j -)}{v_j} \right]^+ = \left[ \sum_{k \in A(v_j
-)}
\frac{B_k^\aug (v_j -)}{v_j} - \frac{B_*^\aug(v_j -)}{v_j} - \sum_{k \in A(v_j)
\setminus i}
\frac{B_k^\aug (v_j -)}{v_j} \right]^+ = \\
& = \left[ \frac{B_i^\aug(v_j -)}{v_j} + \frac{B_j^\aug(v_j -)}{v_j} -
\frac{B_*^\aug(v_j -)}{v_j} \right]^+ = \frac{1}{v_j}\left[ \min\{B_i(0),
B_*^\aug (v_j -)\} + 
\min\{B_j(0), B_*^\aug(v_j -)\} - B_*^\aug(v_j -) \right]^+
\end{aligned}$$}{$$\begin{aligned} \delta_i^\aug & = \left[ S^\aug(v_j -) -
\sum_{k \in A(v_j)
\setminus i} \frac{B_k^\aug (v_j -)}{v_j} \right]^+ = \\ & = \left[ \sum_{k \in
A(v_j
-)}
\frac{B_k^\aug (v_j -)}{v_j} -
\frac{B_*^\aug(v_j -)}{v_j} \right. - \\ & \quad \quad - \left. \sum_{k \i n
A(v_j) \setminus i} \frac{B_k^\aug (v_j -)}{v_j} \right]^+ = \\
& = \left[ \frac{B_i^\aug(v_j -)}{v_j} + \frac{B_j^\aug(v_j -)}{v_j} -
\frac{B_*^\aug(v_j -)}{v_j} \right]^+ = \\ & = \frac{1}{v_j}\left[ \min\{B_i(0),
B_*^\aug (v_j -)\} + \right. \\ & \quad \quad \left. +
\min\{B_j(0), B_*^\aug(v_j -)\} - B_*^\aug(v_j -) \right]^+
\end{aligned}$$}
where the last step is an invocation of Corollary
\ref{cor:budget_profile_format}.
For the base auction we have essentially the same, except that $S^\base(v_j -)
\leq \sum_{k \in A(v_j -)} \frac{B_k^\base (v_j -)}{v_j} - \frac{B_*^\base(v_j
-)}{v_j}$ holds as an inequality rather than equality, so we get:
\columnsversion{$$ \delta_i^\base \leq \left[ \frac{B_i^\base(v_j -)}{v_j} +
\frac{B_j^\base(v_j
-)}{v_j} - \frac{B_*^\base(v_j -)}{v_j} \right]^+ = \frac{1}{v_j}\left[
\min\{B_i(0), B_*^\base (v_j -)\} + 
\min\{B_j(0), B_*^\base (v_j -)\} - B_*^\base (v_j -) \right]^+$$}{
$$ \begin{aligned} \delta_i^\base & \leq \left[ \frac{B_i^\base(v_j -)}{v_j} +
\frac{B_j^\base(v_j
-)}{v_j} - \frac{B_*^\base(v_j -)}{v_j} \right]^+ = \\ & = \frac{1}{v_j}\left[
\min\{B_i(0), B_*^\base (v_j -)\} + \right. \\ & \quad \quad  + \left. 
\min\{B_j(0), B_*^\base (v_j -)\} - B_*^\base (v_j -) \right]^+
\end{aligned}$$}

In order to prove that $B^\aug_i(v_j) \leq B^\base_i(v_j)$, we study two cases:

\begin{itemize}
 \item Case A: $B_*^\aug(v_j -) \leq B_j^\aug(0)$, i.e. $B_j^\aug(v_j -) = 
B_*^\aug(v_j -)$. In this case, $\delta_i^\aug
= \frac{B_i^\aug(v_j -)}{v_j}$ and therefore $B_i^\aug(v_j) = 0$, so, it is
trivial that $B^\aug_i(v_j) = 0 \leq B^\base_i(v_j)$.
 \item Case B: $B_*^\aug(v_j -) > B_j^\aug(0)$. Now, consider the
function $\Phi(\beta) = [\min\{\beta,\mu\} + \min\{\beta, \gamma\} - \beta]^+$ 
for $\beta \geq \min\{\mu,\gamma\}$. This function is monotone non-increasing in
such range. Now, take $\mu = B_i(0), \gamma = B_j(0)$ and use that
$B^\aug_*(v_j -) \leq B^\base_*(v_j -)$ to conclude that $\delta_i^\aug =
\Phi(B_*^\aug(v_j -)) \geq \Phi(B_*^\base(v_j -)) \geq \delta_i^\base$. 
This implies $B^\aug_i(v_j) \leq B_i^\base(v_j)$.
\end{itemize}
This finishes the proof of part (a) of the Meta Lemma.\\

Now, for part (b) of the Meta-Lemma, consider two cases:
\begin{itemize}
 \item $B_*^\aug(p) <
B_*^\base(p)$, then by right-continuity of the budget function, there is
some $\epsilon > 0$ such $B_*^\aug(p') < B_*^\base(p')$ for any $p' \in
[p,p+\epsilon)$.
 \item $B_*^\aug(p) = B_*^\base(p)$,
therefore, $B_i^\aug(p) = B_i^\base(p)$ for all $i \in A(p)$, moreover, $S^\aug
(p) = S^\base(p)$, since
\columnsversion{$$S^\aug(p) = \sum_{i \in A(p)}
\frac{B^\aug_i(p)}{p} -
\frac{B_*^\aug (p)}{p}
= \sum_{i \in A(p)} \frac{B^\base_i(p)}{p} - \frac{B_*^\base (p)}{p} =
S^\base(p)$$}{$$\begin{aligned} S^\aug(p) & = \sum_{i \in A(p)}
\frac{B^\aug_i(p)}{p} -
\frac{B_*^\aug (p)}{p}
= \\ & =  \sum_{i \in A(p)} \frac{B^\base_i(p)}{p} - \frac{B_*^\base (p)}{p} =
S^\base(p) \end{aligned}$$}
Since the behavior of the function $B(\cdot)$ for $p' \geq p$ just depends on
$S(p)$ and $B(p)$, for all $p' \geq p$, then for all $p' \geq p$, $B^\aug(p) =
B^\base(p)$. In other words, when $B_*^\base(p)$ and $B_*^\aug(p)$ meet, then
the auctions become {\em fully coupled}.\\
\end{itemize}

Part (c) of the Meta-Lemma is essentially the same argument made in item (a).
This part is trivial for $p \notin \{v_1, \hdots, v_n\}$ by  continuity of
$B(p)$. For $p = v_j$ we use that $B_*^\aug(v_j -) \leq B_*^\base(v_j -)$ and
study $\delta_i^\aug$ and $\delta_i^\base$. As in (c) we get:

\columnsversion{$$  \delta_i^\aug = \frac{1}{v_j}\left[ \min\{B_i(0),
B_*^\aug (v_j -)\} + 
\min\{B_j(0), B_*^\aug(v_j -)\} - B_*^\aug(v_j -) \right]^+ $$
$$ \delta_i^\base = \frac{1}{v_j}\left[ \min\{B_i(0),
B_*^\base (v_j -)\} + 
\min\{B_j(0), B_*^\base(v_j -)\} - B_*^\base(v_j -) \right]^+ $$}{
$$ \begin{aligned} & \delta_i^\aug = \frac{1}{v_j}\left[ \min\{B_i(0),
B_*^\aug (v_j -)\} + \right. \\  & \quad \quad + \left.
\min\{B_j(0), B_*^\aug(v_j -)\} - B_*^\aug(v_j -) \right]^+ \end{aligned} $$
$$ \begin{aligned} & \delta_i^\base = \frac{1}{v_j}\left[ \min\{B_i(0),
B_*^\base (v_j -)\} + \right. \\  & \quad \quad + \left.	
\min\{B_j(0), B_*^\base(v_j -)\} - B_*^\base(v_j -) \right]^+ \end{aligned} $$}
Now, by analyzing cases A and B as in part (a) of the Meta-Lemma, we conclude
that $B_i^\aug(v_j) = B_i^\aug(v_j -) - v_j \delta_i^\aug \leq B_i^\base(v_j -)
- v_j \delta_i^\base = B_i^\base(v_j)$ as desired.
\end{proofof}

\section{Adaptive Clinching Auction with
Repeated Values}\label{appendix:repeated-values}

In this section, we modify the auction given in Definition
\ref{defn:adaptive_clinching_auction} to account for the possibility of
valuation profiles having repeated values, i.e., two agents $i,i'$ with $v_i =
v_{i'}$. The modification is quite simple:

\begin{defn}[Adaptive Clinching Auction revisited]
 Given any valuation vector $v$, budget vector $B$ and initial supply $s$,
consider functions $x_i(p), B_i(p)$ satisfying (i),(ii) and (iii) in Definition
\ref{defn:adaptive_clinching_auction} and also:
\begin{enumerate}
 \item[(iv')] for $p = v_j -$, Let $\tilde{A} = A(v_j -)$, $\tilde{x}_i =
x_i(v_j -)$, $\tilde{B}_i = B_i(v_j -)$ and $\tilde{S} = S(v_j -)$. Now, run
the following procedure on $(\tilde{A}, \tilde{x}, \tilde{B}, \tilde{S})$:
 \begin{itemize}
  \item[$\circ$] while there is $\tilde{j} \in \tilde{A}$ with
$v_{\tilde{j}} = v_j$
 \begin{itemize}
  \item[$\circ$] let $\tilde{j}$ be the lexicographic first of such elements
  \item[$\circ$] remove $\tilde{j}$ from $\tilde{A}$
  \item[$\circ$] define $\delta_k = \tilde{S} - \sum_{k \in \tilde{A}}
\frac{\tilde{B}_k}{v_j}$ for each $k \in \tilde{A}$
  \item[$\circ$] update $\tilde{x}_k = \tilde{x}_k + \delta_k$, $\tilde{B}_k =
\tilde{B}_k - \delta_k v_j$ for each $k \in \tilde{A}$
  \item[$\circ$] update $\tilde{S} = \tilde{S} - \sum_{k \in \tilde{A}}
\delta_k$.
 \end{itemize}
 \end{itemize}
and then set $x_i(v_j) = \tilde{x}_i$, $B_i(v_j) = \tilde{B}_i$.
\end{enumerate}
\end{defn}

Notice this is essentially repeating (iv) in Definition
\ref{defn:adaptive_clinching_auction} for as many times as elements with the
same value $v_j$. Notice that all the proofs in Section
\ref{sec:multi-unit-auctions} carry out naturally for this setting, simply by
repeating the same argument done for (iv) multiple times, showing that the
invariants analyzed are true after each {\em while} iteration.
\section{Algorithmic Form of the Adaptive Clinching
Auction}\label{appendix:explicit}

We presented the Adaptive Clinching Auction in Definition
\ref{defn:adaptive_clinching_auction} as the limit as $p \rightarrow \infty$
of a differential procedure following Bhattacharya el al
\cite{Bhattacharya10}. Here we present the same auction in an algorithmic
format, i.e., an $\tilde{O}(n)$ steps procedure to compute $(x,\pi)$ from
$(v,B,s)$. The
idea is quite simple: given a price $p$ and the values of $B(p), x(p)$,
we solve the differential equation in item (i) of Definition
\ref{defn:adaptive_clinching_auction} and using it, we compute the next point
$\bar{p}$ where either a player leaves the active set, or a player enters the
clinching set. Given that, we compute $B(\bar{p}-), x(\bar{p}-)$. Then we
obtain the values of $B(\bar{p}), x(\bar{p})$ either by the procedure in (iv)
if a player leaves the active set on $\bar{p}$ or simply by taking $B(\bar{p})
= B(\bar{p}-)$ and $x(\bar{p}) = x(\bar{p}-)$ otherwise.

\begin{lemma}
 Consider the functions $x(p)$ and $B(p)$ obtained in the Adaptive Clinching
Auction. If for prices $p' \in [p, \bar{p})$, the clinching and active set are
the same, i.e., $C(p') = C(p)$ and $A(p') = A(p)$, then given $k = \abs{C(p)}$,
the players $i$ in the clinching set are such that:
\begin{itemize}
 \item if $k=1$, $S(p') = \frac{p S(p)}{p'}$, $x_i(p') = x_i(p) + [S(p')-S(p)]$
and $B_i(p') =  B_i(p) + p S(p) [\log p - \log p']$.
 \item if $k>1$, $S(p') = \frac{p^k S(p)}{(p')^k}$, $x_i(p') = x_i(p) +
\frac{1}{k}[S(p')-S(p)]$  and $B_i(p') = B_i(p) + \frac{p^k
S(p)}{k-1} \left[ \frac{1}{{p'}^{k-1}} - \frac{1}{p^{k-1}} \right]$
\end{itemize}
\end{lemma}

\begin{proof}
 The proof is straightforward. For the case of $k=1$, we
follow the discussion in Example
\ref{defn:adaptive_clinching_auction}: let $i$ be the only player in $C(p)$,
then $S(p') + x_i(p')$ is constant in this range, since all that is subtracted
from the supply is added to the allocation of player $1$, therefore:
$$\partial S(p') = - \partial_p x_i(p') = - \frac{S(p')}{p'} \Rightarrow S(p') =
\frac{\alpha}{p'}.$$
using the boundary condition $S(p) = \frac{\alpha}{p}$, we get the value of
$\alpha = p S(p)$. Now, clearly $x(p') = x(p) + [S(p')-S(p)]$, since player $i$
is the only one clinching. For his budget:
\columnsversion{$$B_i(p') - B_i(p) = \int_p^{p'} \partial_p B_i (\rho) d\rho = 
\int_p^{p'}
-S(\rho) d\rho = \int_p^{p'} -\frac{p S(p)}{\rho} d\rho = p S(p)  [\log p -
\log p']$$}{
$$\begin{aligned} B_i(p') - B_i(p) & = \int_p^{p'} \partial_p B_i (\rho) d\rho =
 \int_p^{p'} -S(\rho) d\rho = \\ & =  \int_p^{p'} -\frac{p S(p)}{\rho} d\rho = p
S(p) [\log p -
\log p'] \end{aligned}$$}
For $k>1$, $S(p') + \sum_{i \in C(p')} x_i(p')$ is constant and therefore:
$$\partial S(p') = - \sum_{i \in C(p') }\partial_p x_i(p') = - k
\frac{S(p')}{p'} \Rightarrow S(p') =
\frac{\alpha}{(p')^k}.$$
Using the boundary condition $S(p) = \frac{\alpha}{p^k}$, we get the value of
$\alpha = p^k S(p)$. We use the observation in Lemma
\ref{lemma:budget_profile_format_pre} that players in the clinching set have
the same budget, and therefore the auction treats them equally from this point
on as long as they remain in the active set, i.e., they will get allocated and
charged at the same rate. Therefore: $x(p') = x(p) + \frac{1}{k}[S(p')-S(p)]$.
For the budgets:
\columnsversion{ $$B_i(p') - B_i(p) = \int_p^{p'}
-S(\rho) d\rho = \int_p^{p'} -\frac{p S(p)}{\rho^k} d\rho = \frac{p^k S(p)}{k-1}
\left[ \frac{1}{(p')^{k-1}} -\frac{1}{p^{k-1}} \right] $$}{
$$\begin{aligned} B_i(p') - B_i(p) & = \int_p^{p'}
-S(\rho) d\rho = \int_p^{p'} -\frac{p S(p)}{\rho^k} d\rho = \\ & =\frac{p^k
S(p)}{k-1} \left[ \frac{1}{(p')^{k-1}} -\frac{1}{p^{k-1}} \right]
\end{aligned}  $$ }
\end{proof}

\begin{theorem}[Algorithmic Form]\label{thm:algorithmic_format}
 It is possible to compute the allocation and payments of the Adaptive
Clinching Auction in $\tilde{O}(n)$ time.
\end{theorem}

\begin{proof}
 Using the lemma above, we just need  to compute $x$ and $B$ for the points where
one of the following events happen: (a) one leaves the active set and (b) one
player enters the clinching set. Clearly there are at most $n$ events of type
(a) and by Lemma \ref{lemma:once_clinching_always_clinching} also at most $n$
events of type (b).

 The algorithm starts at price $p=0$ and at each time computes the next
event. For example, at price $p=0$, the next event of type (a) occurs in
$p=\min_i v_i$. The next event of type (b) occurs at price $p = \frac{1}{s}
[\sum_A B_i - \max_A B_i]$ if no event of type (a) happens before. First we
compute which one occurs first.  Let $\bar{p}$ be such a price. Then, computing
$B(\bar{p}-), x(\bar{p}-)$ is trivial, since no clinching happened so far, so
at that price: $B(\bar{p}-) = B$ (initial budgets) and $x(\bar{p}-) = 0$. Now,
if $\bar{p}$ is an event of type (a), then use step (iv) in Definition
\ref{defn:adaptive_clinching_auction} to compute $x(\bar{p}), B(\bar{p})$. If
not, simply take $B(\bar{p})
= B(\bar{p}-)$ and $x(\bar{p}) = x(\bar{p}-)$.

From this point on, at each considered price $p$, the clinching set will be
non-empty, so we know the format of $x(p')$ and $B(p')$ for $p' \in
[p,p+\epsilon)$. If the next event that happens is of type (a), it happens at
$\min\{v_j; v_j > p\}$, if it is of type (b), it happens at $\min \{p'; B_*(p')
= \max_{i' \in A(p') \setminus C(p')} B_{i'} \}$, where the expression for
$B(p')$ is given in the previous lemma. For example, if $\abs{C(p)} = 1$, then
this happens at:
$$p' = \exp\left[ \frac{1}{p S(p)} (B_*(p) - \textstyle\max_{i' \in
A(p)\setminus C(p)} B_{i'}(p) ) + \log p\right]$$
and if $\abs{C(p)} = k > 1$, it happens at:
$$p' = \frac{p^k S(p)}{(k-1)(B_*(p) - \textstyle\max_{i' \in A(p)\setminus
C(p)} B_{i'}(p)) + p S(p)}$$
Those expressions are easily obtained by taking $B_i(p')$ as calculated in the
previous lemma and calculating for which $p'$ it becomes equal to $\max_{i' \in
A(p) \setminus C(p)} B_i'(p)$.
Now, we simply need to find out which of those events happen first. Let it be
$\bar{p}$, then we compute $B(\bar{p}-), x(\bar{p}-)$ using the previous lemma
and then update to $B(\bar{p}), x(\bar{p})$ as described above.
\end{proof}

\section{Details of the Qualitative Description}\label{appendix:closed_form_2}

In Theorem \ref{thm:main_theorem_multi_units}, we showed that the Adaptive
Clinching Auction is an auction in the online supply model. It is natural
to ask how allocation and payment qualitatively evolve as supply arrives. In
order words, how does it allocate and charge for an $\epsilon$ amount of the
good when it arrives after $s$ supply has already been allocated. We perform
in this section a qualitative analysis for two players based on the explicit
formula of the Adaptive Clinching Auction for $n=2$ players derived in
Dobzinski et al \cite{dobzinski12} (which can  alternatively be obtained from
Theorem \ref{thm:algorithmic_format}). This is reproduced in Figure 
\ref{fig:explicit_ACA_n2}.

As we will see, the clinching auction has a very natural behavior for the first
few units of supply that arrive - it simply allocates them using VCG. At a
certain point, when budget constraints start to kick-in, the allocation and
payments evolve in a quite non-expected way. The main goal of this section is
to highlight this point.

Depending on the relation between
$v_1$ and $v_2$, two distinct behaviors can happen: either at a certain point,
the high-value player gets his budget depleted, and the auction starts
allocating new arriving units to the low-value player (still charging for
those items) and then at a certain
further point, it starts splitting the goods among them (charging only the player
with non-depleted budget). Or alternatively, the
auction continues to allocate to the high-value player, but start charging him
a discounted version of VCG. Then when his budget gets depleted, the auction
starts splitting goods among the players (charging only the player with
non-depleted budget).

\columnsversion{\begin{figure}[h]}{\begin{figure*}[h]}
\centering
\begin{tabular}{|c|l|}
\hline
\multirow{2}{*}{ $v_2 \geq v_1 \text{ and } s \cdot v_1 \leq B_2$}
& $x = \left(0, s\right)$ \\
& $p = \left(0, s v_1 \right)$ \\
\hline
\multirow{2}{*}{ $v_2 \geq v_1 \text{ and } B_2 \leq s \cdot v_1 \leq B'_1$}
& $x = \left(s-\frac{B_2}{v_1}, \frac{B_2}{v_1} \right)$ \\
& $p = \left(  B_2 (\log (s v_1) - \log B_2) , B_2\right)$
\\
\hline
\multirow{2}{*}{ $v_2 \geq v_1 \text{ and } B'_1 \leq s \cdot v_1 $}
& $x = \left( s-\frac{s B_2}{2 B'_1} [ 1+ ( \frac{B'_1}{s v_1} )^2
],\frac{s B_2}{2 B'_1} [ 1+ ( \frac{B'_1}{s v_1} )^2
] \right)$ \\
& $p = \left( B_2 (1-\frac{B'_1}{s v_1}) + B_2 (\log B'_1 - \log B_2) ,
B_2\right)$
\\
\hline
\multirow{2}{*}{ $v_2 < v_1 \text{ and }  s \cdot v_2 \leq B_2$}
& $x = \left(s, 0\right)$ \\
& $p = \left(s v_2, 0 \right)$ \\
\hline
\multirow{2}{*}{ $v_2 < v_1 \text{ and } B_2 \leq s \cdot v_2 \leq B'_1$}
& $x = \left(s, 0\right)$ \\
& $p = \left(B_2 + B_2 (\log (s v_2) - \log B_2), 0 \right)$ \\
\hline
\multirow{2}{*}{ $v_2 < v_1 \text{ and } B'_1 \leq s \cdot v_2$}
& $x = \left( s+\frac{ s B_2}{2 B'_1} [-1+(\frac{B'_1}{s v_2})^2], \frac{s
B_2}{2 B'_1}
[1-(\frac{B'_1}{s v_2})^2] \right)$ \\
& $p = \left(B_1, B_2  - \frac{B'_1 B_2}{s v_2}  \right)$ \\
\hline 
\end{tabular}\\
\caption{Explicit formula of the Adaptive Clinching Auction for
$n=2$. \\
Valuations $v_1, v_2$, budgets $B_1 \geq B_2$, initial supply $s$ and $B'_1 =
\exp(\frac{B_1}{B_2}-1 + \log B_2)$ }
\label{fig:explicit_ACA_n2}
\columnsversion{\end{figure}}{\end{figure*}}

\columnsversion{\begin{figure}}{\begin{figure*}}
\centering
\includegraphics{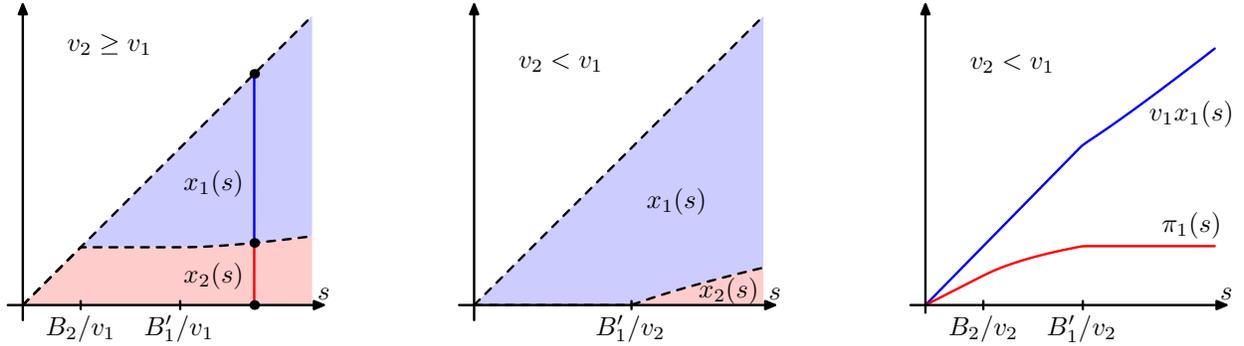}
\caption{Evolution of allocations $x_1$ and $x_2$ as supply arrives in the
cases $v_2 \geq v_1$ (first plot) \\ and $v_2 < v_1$ (second plot). The third
plot represents the evolution of value $v_1 x_1(s)$ and \\ payment $\pi_1(s)$ of
player $1$ as more supply arrives in the case $v_2 < v_1$.}
\label{fig:interpretation_1}
\columnsversion{\end{figure}}{\end{figure*}}

Now, for some fixed pair of valuations $v_1, v_2$ and budgets $B_1 \geq B_2$, we
study how the allocation and payments evolve with the available supply $s$. As
one can easily see, if
supply is sufficiently small $s \leq \frac{\min_i B_i}{\min_i v_1}$, the
auction is essentially VCG, as the good arrives, it  allocates it to the highest-value 
player and she pays per unit the value of the other player. The behavior
afterwards depends on the relation between $v_1$ and $v_2$.

\textbf{Case $\mathbf{v_2 \geq v_1}$ :} As goods arrive, we allocate to the
higher-value player charging him $v_1$ per unit. This is essentially VCG. This
is possible until $s = \frac{B_2}{v_1}$, when the budget of $2$ is depleted.
As more goods arrive, we allocate them entirely to player $1$ charging them at
a rate proportionally to $\frac{B_2}{s}$, i.e., the fraction between player
$2$'s original budget and supply that has arrived so far. We continue doing
that until $s = \frac{1}{v_1} B'_1 = \frac{1}{v_1}\exp(\frac{B_1}{B_2}-1+\log
B_2)$. At this point, the remaining budget of player $1$ is the same as the
original budget of player $1$. From this point on, each amount of the good that
arrives is split among $1$ and $2$ at a rate $\partial_s x = (1-\frac{B_2}{2
B'_1} + \frac{B_2 B'_1}{2 v_1^2 s^2},\frac{B_2}{2
B'_1} - \frac{B_2 B'_1}{2 v_1^2 s^2})$. Player $2$ is clearly not charged
(since his budget is already depleted) and player $1$ is charged at a rate
$\frac{B'_1 B_2}{s^2 v_1}$ per arriving unit of $s$. (The allocation is
depicted in the first part of Figure \ref{fig:interpretation_1}.)

\textbf{Case $\mathbf{v_2 < v_1}$ :} Again we start allocating like VCG, i.e.,
allocating the goods to player $1$ and charging him $v_2$ for each unit of the
good. We do that up to supply $s = \frac{B_2}{v_2}$. From this point on, we
continue allocating incoming goods to player $1$, but we charge him at a
cheaper rate than $v_2$, precisely, we charge him at a rate $\frac{B_2}{s}$. We
do so until the budget of $1$ is depleted, which happens at $s =
\frac{B'_1}{v_2}$. From this point on, we split the arriving goods between
players $1$ and $2$ at a rate $\partial_s x = (1-\frac{B_2}{2
B'_1} - \frac{B_2 B'_1}{2 v_1^2 s^2},\frac{B_2}{2
B'_1} + \frac{B_2 B'_1}{2 v_1^2 s^2})$. Naturally, we cannot charge player $1$,
because his budget is already depleted, but we charge player $2$ at a rate
$\frac{B'_1 B_2}{s^2 v_2}$ per arriving unit of $s$. (The allocation is
depicted in the second and third part of Figure \ref{fig:interpretation_1}.)\\

\paragraph{Relation to bid throttling.} One important remark is that the auction
is {\em not} a special case of a bid-throttling scheme, since an agent is
allocated items even after his budget is completely depleted and even if the
other agent still has budget left. Intuitively, the fact that an agent got many
items for an expensive price in the beginning (once the items were scarce) gives
him an advantage over items in the future if/once they become abundant, i.e.,
this agent will have the possibility of acquiring these items for a lower price
compared to other agents.

The final goal of this research direction is to provide better simple heuristics
to deal with budgets for real-world ad auctions. We believe that the
qualitative analysis
above hints to new heuristics to manage budget constrained agents.
We illustrate the effectiveness of such heuristic in the following scenario:
consider a set of advertisers competing for ad slots on queries for
a highly volatile query, say `sunglasses'.  If weather is rainy, there
are very few queries and budget constraints do not kick in, but if weather is
sunny, there is a high volume of queries and we would like to split the queries
to advertisers according to their budgets. Imagine now that the weather is
completely unpredictable. If the day starts rainy, very few queries arrive
in the morning. It is unclear is the weather is changing in the afternoon.
If the search engine knew the weather in the afternoon, it could run second-price 
auctions on modified bids to ensure that the high value players still
have sufficient budget left in the afternoon if it is sunny and would run
second price on real bids if they knew the weather would still be rainy.
The heuristic proposed by clinching on the other hand, is completely agnostic
to that matter. It will allocate using VCG in the beginning. If the weather
becomes sunny, it will  give items for cheaper for the high-valued players
to compensate for the more expensive items acquired in the beginning.

\end{document}